\newcommand{\mpar}[2]{\marginpar{\raggedright\tiny#2~[#1]}}
\renewcommand{\mpar}[2]{}

\documentclass[11pt]{amsart}
\usepackage{latexsym}
\usepackage{amssymb}
\usepackage{amsmath}

\newtheorem{theorem}{Theorem}[section]
\newtheorem{lemma}[theorem]{Lemma}
\newtheorem{proposition}[theorem]{Proposition}
\newtheorem{corollary}[theorem]{Corollary}

\theoremstyle{definition}
\newtheorem{example}[theorem]{Example}
\theoremstyle{definition}
\newtheorem{definition}[theorem]{Definition}
\theoremstyle{definition}
\newtheorem{remark}[theorem]{Remark}
\newtheorem*{remark*}{Remark}

\newcommand\MD{\mathop{\rm MD}\nolimits}
\newcommand\GL{\mathop{\rm GL}}
\newcommand\Sp{\mathop{\rm Sp}}
\newcommand\id{\mathop{\rm id}\nolimits}
\newcommand\Tr{\mathop{\rm Tr}\nolimits}

\let\epsilon\varepsilon

\newcommand\R{\mathbb{R}}
\newcommand{\lan}{\langle}
\newcommand{\ra}{\rangle}
\newcommand{\ten}{\otimes}
\newcommand{\vphi}{\varphi}
\newcommand{\oten}{\bar{\otimes}}
\newcommand{\ep}{\varepsilon}
\newcommand{\BLTRN}[1][n]{\mathcal{B}(L_2(\mathbb{R}^{#1}))}
\newcommand{\TLTRN}[1][n]{\mathcal{T}(L_2(\mathbb{R}^{#1}))}
\newcommand{\LTRN}[1][n]{L_2(\mathbb{R}^{#1})}
\newcommand{\LTR}{L_2(\mathbb{R})}

\newcommand{\Norm}[1]{\lVert#1\rVert}
\newcommand{\cbnorm}[1]{\lVert#1\rVert_{\cb}}
\newcommand{\hs}{\hskip2pt}

\newcommand{\smallmat}[1]{\left(\begin{smallmatrix}#1\end{smallmatrix}\right)}

\newcommand{\qand}{\quad\text{and}\quad}

\newcommand\cb{\mathop{\rm cb}}

\newcommand{\cl}[1]{\mathcal{#1}}
\newcommand{\bb}[1]{\mathbb{#1}}

\begin{document}

\title[Private Algebras And Infinite-Dimensional Complementarity]{Private Algebras In Quantum Information And Infinite-Dimensional Complementarity}

\mpar{v5}{Changed title: thoughts?}
\mpar{v6}{\vspace{10pt} Looks good. Shall we say \lq quantum information theory'?}
\mpar{v7}{Done, with short version for headings}
\mpar{v9}{Changed title. I think it's important we have both ''private algebras'' and ''complementarity'' in there somewhere.}

\author[J. Crann]{Jason Crann}
\address{$^1$School of Mathematics \& Statistics, Carleton University, Ottawa, ON, Canada H1S 5B6}
\address{$^2$Universit\'{e} Lille 1 - Sciences et Technologies, UFR de Math\'{e}matiques, Laboratoire de Math\'{e}matiques Paul Painlev\'{e}
- UMR CNRS 8524, 59655 Villeneuve d'Ascq C\'{e}dex, France}
\email{jason\_crann@carleton.ca}

\author[D. W. Kribs]{David W. Kribs}
\address{$^1$Department of Mathematics \& Statistics, University of Guelph, Guelph, ON, Canada N1G 2W1}
\address{$^2$Institute for Quantum Computing, University of Waterloo, Waterloo, ON, Canada N2L 3G1}
\email{dkribs@uoguelph.ca}

\author[R. H. Levene]{Rupert H. Levene}
\address{School of Mathematical Sciences, University College Dublin, Belfield, Dublin 4, Ireland}
\email{rupert.levene@ucd.ie}

\author[I. G. Todorov]{Ivan G. Todorov}
\address{Pure Mathematics Research Centre, Queen's University Belfast, Belfast BT7 1NN, United Kingdom}
\email{i.todorov@qub.ac.uk}

\date{14 August 2015}

\maketitle

\begin{abstract}
  We introduce a generalized framework for private quantum codes using
  von Neumann algebras and the structure of commutants. This leads
  naturally to a more general notion of complementary channel, which
  we use to establish a generalized complementarity theorem between
  private and correctable subalgebras that applies to both the finite and
	infinite-dimensional settings. Linear bosonic channels are considered and
	specific examples of Gaussian quantum channels are given to illustrate the new framework together
  with the complementarity theorem.
\end{abstract}

\mpar{v9}{Added mention in Abstract of the comp thm applying to both fin and inf cases.}	

\section{Introduction}

One of the most basic notions in quantum privacy is the private quantum code. Arising initially as the quantum
analogue of the classical one-time pad, they were first called private quantum channels and investigated for optimal encryption
schemes \cite{amtdw,br03}. The subject has grown considerably over the past decade and a half, with related applications in quantum secret
sharing \cite{cgs,cgl} and the terminology ``private quantum subsystems'' taking hold as part of work on the theory of private shared
reference frames \cite{bhs,brs}. In recent years, focus in the subject has turned to investigating relevant properties of completely
positive maps. This has led to connections established with quantum error correction \cite{kks}, discussed in more detail below, as well as
algebraic conditions characterizing private subsystems and new, surprisingly simple examples that suggest private subsystems are more
ubiquitous than previously thought \cite{jklp,jklp1}. These more recent works, along with \cite{ckpp}, have also suggested deeper connections with the theory of
operator algebras, opening up the possibility of extending the subject to infinite-dimensional Hilbert spaces.

\mpar{v9}{Reworked introduction, including added 1st paragraph.}	

From a different but related direction, throughout the development of quantum theory, the notion of complementarity
has played a fundamental role in the interpretation of quantum measurements,
providing, for instance, the theoretical basis behind quantum state tomography. At the level of
quantum channels, an appropriate notion of complementarity has been formulated and shown to be vital for understanding
their overall structure \cite{hol05,kmnr}. An underlying feature of complementarity is the trade-off between
information and disturbance. For finite-dimensional quantum channels,
this trade-off was quantified in \cite{ksw}, and was used to establish
a complementarity theorem between private and correctable
subsystems for a channel and its complementary channel \cite{kks}.

As there is a more general framework for (infinite-dimensional)
quantum error correction at the level of von Neumann algebras \cite{b,bkk1,bkk1a,bkk2},
a natural question is to seek a generalized notion of private quantum codes that is also viable in the infinite-dimensional
setting, and for which a suitable complementarity theorem holds. Using von Neumann algebras
and the structure of commutants, in this paper we introduce a generalized framework
for private quantum codes which may be seen as the complementary analogue of so-called operator
algebra error correction, resulting in a natural notion of ``private algebras''. This in turn leads to a more
general notion of complementary channel, and we establish a generalized complementarity
theorem for arbitrary dimensions in the new framework. As a corollary, we also obtain
a structure theorem for correctable subalgebras that generalizes a
finite-dimensional result \cite{jk}. We finish by illustrating the framework and concepts for infinite-dimensional
linear bosonic channels and a specific class of Gaussian quantum channels \cite{h2,h}.

The outline of the paper is as follows. In Section \ref{s_prel}, we
discuss the necessary preliminaries on infinite-dimensional channels
and von Neumann algebras. We then introduce our
generalized framework for private quantum codes in Section \ref{s_pacs},
and discuss their basic properties and examples. Section \ref{s_ccs} contains the
generalized complementarity theorem and its aforementioned application.
In Section \ref{s_palbc}, we study
explicit examples of linear bosonic and Gaussian quantum channels
which illustrate the new framework along with the complementarity
theorem. We end with a conclusion summarizing the results of the paper
and an outlook on future work.

\section{Preliminaries}\label{s_prel}

Let $S$ be a (not necessarily finite-dimensional) Hilbert space.  We
assume that the inner product is linear in the second variable
and denote by $\cl B(S)$ (resp. $\cl T(S)$) the space of all bounded linear
(resp. trace class) operators on $S$.
There is a canonical isometric isomorphism between the Banach space
dual $\cl T(S)^*$ of $\cl T(S)$ and $\cl B(S)$ via the trace:
$$\lan T,\rho\ra:=\Tr(T\rho), \quad T\in \cl B(S), \rho\in\cl T(S).$$
Thus, $\cl T(S)$ can be identified with the space of normal ({\it i.e.}~weak* continuous) linear
functionals on $\cl B(S)$, where,
if $|\eta\rangle\in S$ and $\langle \xi|$ belongs to the dual $S^*$ of $S$, the rank one operator
$|\eta \rangle \langle \xi|\in \cl T(S)$ corresponds to the vector functional given by
$\omega_{\xi,\eta}(X) = \langle \xi | X | \eta\rangle$, $X\in \cl B(S)$.

We denote by $\cl S(S)$ the set of all \emph{states} on $S$; thus, an
element $\rho\in \cl T(S)$ belongs to $\cl S(S)$ precisely when $\rho$
is positive (that is, $\lan X,\rho\ra=\Tr(X\rho) \geq 0$ whenever
$X\geq 0$) and $\lan I,\rho\ra=\Tr(\rho) = 1$.

If $S$ and $S'$ denote the respective input and output systems of a dynamical quantum process, then,
in the Schr\"{o}dinger picture, states in $\cl T(S)$ evolve under a completely positive trace preserving (CPTP) map to states
in $\cl T(S')$. In the Heisenberg picture, which will be adopted in this paper,
observables in $\cl B(S')$ evolve under a normal
({\it i.e.}~weak*-weak* continuous) unital completely positive (NUCP)
\mpar{v6}{Def. of quantum channel not given at this point -- done below more generally.}
map $\cl E$ to observables in $\cl B(S)$.
As a normal map, $\cl E$ has a unique pre-adjoint $\cl E_*:\cl T(S)\rightarrow\cl T(S')$ which is a
CPTP map describing the corresponding evolution of states.

Suppose that in the above scenario, one wished, or had the ability, to
measure only a certain subset $\cl O$ of observables on the output
space $S'$.  The results of the measurements will then be governed by
the spectral projections of the corresponding elements in $\cl O$,
which, by general spectral theory, lie in the von Neumann algebra $M$
generated by $\cl O$.  Thus, the relevant dynamics is encoded in the
restriction of $\cl E$ to $M$, that is, in a NUCP map $\cl
E:M\rightarrow\cl B(S)$.  As such mappings are natural objects of
study in operator algebra theory, and include the class of classical
information channels, we will adopt this more general framework in
this paper. The remainder of this section will be devoted to a
brief overview of the relevant concepts; for details, we refer the
reader to~\cite{Pa,t1}.\mpar{v10}{ref corrected}

A \emph{von Neumann algebra} on a Hilbert space $S$ is a
$*$-subalgebra $M$ of $\cl B(S)$ with unit~$1_M=I_S\in M$
\mpar{v7}{Slightly ugly rephrasing, but we are using both $1_M$ and
  $I_S$ which could confuse some readers. I have tried to make our
  usage more consistent below.}%
 which is closed in the strong operator topology.
For a subset $L\subseteq\cl B(S)$, its \emph{commutant} is the
subspace
$$L':=\{X\in\cl B(S)\mid XT=TX, \  \mbox{ for all } T\in L\}.$$
Von Neumann's bicommutant theorem states that a unital $*$-subalgebra
$M$ of $\cl B(S)$ is a von Neumann algebra if and only if $M'':=(M')'$
coincides with~$M$.  As $(L')''=L'$ for any subset
$L\subseteq\cl B(S)$, the commutant $M'$ of a von Neumann algebra
$M$ is again a von Neumann algebra on $S$.

Another distinguishing feature of a von Neumann algebra $M$ is that it is (isometrically isomorphic to) the dual of a unique Banach space $M_*$,
called the \emph{predual} of $M$, which consists of all weak* continuous linear functionals on $M$. For example, $M=\cl B(S)$ is a von Neumann algebra with $M_*=\cl T(S)$.
We will denote by $\cl S(M)$ the set of normal states on $M$, that is,
the positive elements $\rho$ of $M_*$ satisfying $\lan I_S,\rho\ra=1$.
If $M$ and $N$ are von Neumann algebras,
a bounded linear map $\cl E : M\rightarrow N$ is said to be \emph{normal} if it is weak*-weak* continuous.
In this case, $\cl E$ has a unique pre-adjoint $\cl E_*:N_*\rightarrow M_*$ satisfying
$$\lan X,\cl E_*(\rho)\ra=\lan\cl E(X),\rho\ra, \quad X\in M,\  \rho\in N_*.$$
Moreover, $\cl E$ is a NUCP map if and only if $\cl E_*$ is completely positive and
$\cl E_*(\cl S(N))\subseteq\cl S(M)$.

Given two Hilbert spaces $S$ and $S'$, we denote by $S\otimes S'$
their Hilbertian tensor product.  For operators $X\in \cl B(S)$ and
$Y\in \cl B(S')$, as usual we denote by $X\otimes Y$ the (unique)
operator in $\cl B(S\otimes S')$ with $(X\otimes Y)(\xi\otimes\eta) =
X\xi\otimes Y\eta$, $\xi\in S$, $\eta\in S'$.  If $M\subseteq\cl B(S)$
and $N\subseteq\cl B(S')$ are von Neumann algebras, the weak* closed
linear span $M\bar{\otimes} N$ of $\{X\otimes Y\mid X\in M,\ Y\in N\}$
is a von Neumann subalgebra of $\cl B(S\otimes S')$.  In particular,
$\cl B(S)\bar\otimes\cl B(S') = \cl B(S\otimes S')$.  If $\rho\in M_*$
and $\omega\in N_*$, then there exists a (unique) element
$\rho\otimes\omega\in (M\oten N)_*$ such that
$$\lan X\ten Y, \rho\otimes\omega\ra = \lan X,\rho\ra\lan Y,\omega\ra, \quad X\in M,\ Y\in N.$$
Thus, we have a natural embedding of the algebraic tensor product $M_*\odot N_*$
into $(M\oten N)_*$; its image is norm dense in $(M\oten N)_*$.

Given a Hilbert space $S$ and a von Neumann algebra $M$, a \emph{quantum channel} is a NUCP map
$\cl E : M\to \cl B(S)$. (This is the dual viewpoint of how channels are typically presented
in quantum information theory as CP trace-preserving maps, but is more natural in the operator algebra
setting.)
Note that a quantum channel $\cl E$ is automatically completely bounded (see {\it e.g.}~\cite{Pa}).
We denote by $\|\Phi\|_{\cb}$ the c.b.~norm of a completely bounded map $\Phi$.
\mpar{v9}{Added note on this being the dual of how usually presented in the QI literature.}
\mpar{v6}{Notation for cb norm included.}
In the case $M = \bb{C}$, an important example is the depolarizing channel
\mpar{v4}{right terminology?}%
\mpar{v5}{Standard terminology for its predual}%
$\cl D : \bb{C}\to \cl B(S)$ given by $\cl D(\lambda) = \lambda I$.
It is straightforward to check that $\cl D_* : \cl T(S)\to \bb{C}$ coincides with the trace.
If $\cl F:N\rightarrow\cl B(S')$ is another quantum channel on the von Neumann algebra $N$, then there exists a (unique)
quantum channel $\cl E\otimes \cl F : M\oten N\rightarrow\cl B(S\ten S')$ such that
$$(\cl E\otimes \cl F)(X\ten Y)=\cl E(X)\ten\cl F(Y),\quad X\in M,\ Y\in N.$$
Channels can similarly be tensored in the Schr\"odinger picture, and
$(\cl E\otimes \cl F)_* = \cl E_*\otimes \cl F_*$.
\mpar{v3}{Sentence added.}

Stinespring's theorem for normal maps
\mpar{v4}{Is there a handy ``Stinespring for normal maps'' reference?}%
\mpar{v5}{Not sure}
asserts that if $\cl E:M\rightarrow\cl B(S)$ is a quantum
channel, then there exist a Hilbert space $H$, a normal unital
\mpar{v4}{added ``unital''}%
$*$-homomorphism $\pi:M\rightarrow\cl B(H)$ and an isometry $V :
S\rightarrow H$ such that
\begin{equation}\label{eq_srep}
\cl E(X) = V^*\pi(X)V, \quad X\in M.
\end{equation}
We refer to the triple $(\pi,V,H)$ as a \emph{Stinespring triple} for
$\cl E$, and to \mpar{v3}{term Stinespring triple introduced} identity
(\ref{eq_srep}) as a \emph{Stinespring representation} of $\cl E$.
Such a Stinespring representation is unique up to
a conjugation by a partial isometry
in the following sense: if $(\pi_1,V_1,H_1)$ and $(\pi_2,V_2,H_2)$ are
Stinespring triples for $\cl E$, then there is a partial isometry
$U:H_1\rightarrow H_2$ such that
\begin{equation}\label{Stine}
UV_1=V_2, \quad U^*V_2=V_1 \ \ \text{and} \ \ U\pi_1(X)=\pi_2(X)U\end{equation}
for all $X\in M$. If~$(\pi_1,V_1,H_1)$ yields a \emph{minimal}
Stinespring representation, meaning that the linear span of
$\pi_1(M)V_1S$ is a dense subspace of $H_1$, then we will call
$(\pi_1,V_1,H_1)$ a \emph{minimal Stinespring triple} for $\cl
E$.\mpar{v4}{minimal triple defined} In this case, the map $U$ above
is necessarily an isometry,
\mpar{v4}{isometry comment added}
and any two minimal Stinespring representations for $\cl E$ are
unitarily equivalent.

\section{Private Quantum Codes via Commutant Structures}\label{s_pacs}
\mpar{v3}{Improve title of section}
\mpar{v5}{Changed: thoughts?}
\mpar{v6}{Looks good to me.}

We now introduce our generalized notion of privacy for quantum channels.
Given Hilbert spaces $S$ and $S'$ and a bounded operator $T : S'\to S$,
we write $\cl C_T : \cl B(S')\to \cl B(S)$, $\cl C_T(X) = TXT^*$ for
conjugation by $T$. \mpar{v10}{slight rephrase}%
Clearly, if $T$ is a partial isometry then $\cl C_T$ is a
quantum channel from $\cl B(S')$ into $\cl B(TT^*S)$.
\mpar{v3}{New notation introduced for conjugation and used throughout the paper.}%
%
%
For a Hilbert space $S$, we let $\cl P(S)$ denote the set of projections in $\cl B(S)$.


\begin{definition}
Let $S$ be a Hilbert space, let $M$ be a von Neumann algebra, and
let $\cl E : M\rightarrow\cl B(S)$ be a quantum channel.
If $P\in\cl P(S)$, a von Neumann subalgebra
\mpar{v3}{terminology slightly changed}%
$N\subseteq\cl B(PS)$ is called \emph{private for $\cl E$ with respect to $P$}
if\mpar{v10}{display important definition}
\[\cl C_P\circ\cl E(M)\subseteq N'.\] Given $\varepsilon>0$,
we say that $N$ is $\varepsilon$-\emph{private for $\cl E$ with respect to $P$}
if there exists a quantum channel $\cl F:M\rightarrow\cl B(S)$ such that
$$\Norm{\cl E-\cl F}_{\cb}<\varepsilon$$
and $N$ is private for $\cl F$ with respect to $P$. If $P = I$,
we simply say that $N\subseteq\cl B(S)$ is
\emph{private} (resp.\ $\varepsilon$-\emph{private}) for $\cl E$.
\end{definition}

\begin{remark*}
  The definition of a private subalgebra is motivated by the notion of
  an operator private subsystem \cite{bhs,brs,jklp,jklp1,kks}.
	\mpar{v9}{Added two references. ''Operator private'' terminology was actually introduced in \cite{jklp,jklp1} but was
	the focus of the examples in the other three papers, and is a special case of the general private notion that goes back to 2000. I
	can give more historical context when we talk.}%
	Recall that, if
  $S, A, B$ and $S'$ are finite-dimensional Hilbert spaces with $S = (A\ten
  B)\oplus(A\ten B)^{\perp}$ and $\cl E : \cl B(S')\rightarrow\cl
  B(S)$ is a UCP map with pre-adjoint $\cl E_*:\cl T(S)\rightarrow\cl
  T(S')$, then $B$ is called an operator private subsystem for $\cl E$
  if $\cl E_*\circ (\cl C_P)_* = \cl F_*\ten \Tr$ for some quantum
  channel $\cl F:\cl B(S')\rightarrow \cl B(A)$, where $P$ is the
  projection from $S$ onto $A\ten B$ \cite{kks}.  Assuming $P\rho P =
  \sum_{i=1}^n \rho_i^A\ten\rho_i^B$, where $\rho_i^A$
  (resp.\ $\rho_i^B$) are elements of $\cl T(A)$ (resp.\ $\cl T(B)$), we
  have
  \begin{align*}\lan\cl C_P\circ\cl E(T),\rho\ra&=\lan T, \cl E_*\circ
    (\cl C_P)_*(\rho)\ra
    = \lan T,\cl E_*(P\rho P)\ra\\
    & = \sum_{i=1}^n\lan T,\cl E_* (\rho_i^A\ten\rho_i^B)\ra
    = \sum_{i=1}^n\lan T, (\cl F_*\ten\Tr)(\rho_i^A\ten\rho_i^B)\ra\\
    & =\sum_{i=1}^n\lan T, \cl F_*(\rho_i^A)\ra\lan I_B,\rho_i^B\ra
    = \sum_{i=1}^n\lan \cl F(T)\ten I_B,\rho_i^A\ten\rho_i^B\ra\\
    & = \lan \cl F(T)\ten I_B,P\rho P\ra = \lan P(\cl F(T)\ten
    I_B)P,\rho\ra.\end{align*}
     \mpar{v3}{an
    extra sentence added}
  \mpar{v4}{Do we need to invoke continuity in the finite dimensional case?}
  \mpar{v5}{No. Removed}
Thus,
$$\cl C_P\circ\cl E(\cl B(S'))\subseteq \cl B(A)\ten I_B = (I_A\ten\cl B(B))' = N',$$
where $N:=I_A\ten\cl B(B) = \{I_A\otimes Y : Y\in \cl B(B)\}$, a
von Neumann subalgebra of $\cl B(PS)$.

Conversely, if $\cl C_P\circ \cl E(\cl B(S'))\subseteq N' = \cl
B(A)\ten I_B$, then for every $T\in\cl B(S')$ there exists $X_T\in\cl
B(A)$ such that $\cl C_P\circ\cl E(T) = X_T\ten I_B$. Set $\cl F(T) =
X_T$.  It is easy to check that this defines a UCP map $\cl F\colon\cl B(S')\to \cl B(A)$
and that its pre-adjoint
$\cl F_*:\cl T(A)\rightarrow T(S')$ satisfies $\cl E_*\circ (\cl C_P)_* = \cl
F_*\ten \Tr$.  Thus, $B$ is an operator private subsystem if and only
if $N\cong\cl B(B)$ is a private subalgebra for $\cl E$ with respect
to $P$.

The choice of the term \lq\lq private'' is justified by
the fact that any information stored in the operator private subsystem
$B$ completely decoheres under the action of $\cl E_*$ \cite{amtdw,brs}.
From the Heisenberg perspective, observables on the output system evolve under $\cl E$ to observables having uniform statistics with respect to the subsystem $B$
in the sense that the expected value of a measurement of $\cl E(T)$ in the state $\rho\in\cl T(A\ten B)$ solely depends on the marginal state $\Tr_B(\rho)\in\cl T(A)$.

In the more general setting of private subalgebras, not all information about observables in the subalgebra $N\subseteq\cl B(PS)$ is lost under the action of
$\cl E:M\rightarrow\cl B(S)$, just the \emph{quantum} information. Indeed, the only obtainable information about $N$ after an application of the channel is the classical information
contained in its center $\cl Z(N)=N\cap N'$. Thus, we recover the usual sense of privacy when $N$ is a \emph{factor}, meaning $\cl Z(N)=\mathbb{C}I$.
If $N$ is a factor of type~I, then $N\cong I_A\ten \cl B(B)$ for some Hilbert spaces $A$ and $B$ \cite{t1}.
This induces a decomposition $S=(A\ten B)\oplus(A\ten B)^{\perp}$ and it follows that $B$ is an operator private subsystem for $\cl E$.
Hence, operator private subsystems are precisely the private type~I factors.
\end{remark*}


\noindent {\bf Examples.}
An immediate class of examples of private subalgebras arises from normal conditional expectations.
If $S$ is a Hilbert space and $\cl E:\cl B(S)\rightarrow N'$ is a normal conditional expectation, that is, a weak*-weak* continuous projection of norm one,
where $N\subseteq\cl B(S)$ is a von Neumann subalgebra,
then trivially, $N$ is private for the quantum channel $\cl E$. Some concrete examples are the following.

\smallskip

\noindent {\bf (i)}
{\it Deletion channels:} $\cl E(T)=\lan T,\rho\ra I$, for some $\rho\in\cl S(S)$; in this case $N=\cl B(S)$.

\smallskip

\noindent {\bf (ii)}
{\it Uniform phase-flips on $n$-qubits:}
$$\cl E(T)=\frac{1}{2^n}\sum_{(s_1,\cdots,s_n)\in\bb{Z}_2^n}Z_{(s_1,\cdots,s_n)}TZ_{(s_1,\cdots,s_n)}^*,$$
where $Z_{(s_1,\cdots,s_n)}=\ten_{i=1}^nZ_{s_i}$ where $Z_0 = I$
and $Z_1 = \smallmat{1&0 \\ 0&-1}$; in this case, $N=N'=\ten_{i=1}^n\Delta_2$,
where $\Delta_2$ is the diagonal subalgebra of $M_2(\bb{C})$.

\smallskip

\noindent {\bf (iii)}
{\it Uniform bit-flips on $n$-qubits:}
$$\cl E(T)=\frac{1}{2^n}\sum_{(s_1,\cdots,s_n)\in\bb{Z}_2^n}X_{(s_1,\cdots,s_n)}T X_{(s_1,\cdots,s_n)}^*,$$
where $X_{(s_1,\cdots,s_n)}=\ten_{i=1}^nX_{s_i}$ with $X_0 = I$
and $X_1 = \smallmat{0&1 \\ 1&0}$; in this case, $N=N'=\ten_{i=1}^n\cl C_2$,
where $\cl C_2$ is the subalgebra of circulant matrices in $M_2(\bb{C})$.

\smallskip

\mpar{v5}{Should we shorten this list of examples? Perhaps only the first and last?}%
\mpar{v7}{Seems fine as is to me.}%
\mpar{v8}{Another vote for keeping the examples as they are.}

The latter two examples fall under a general class of conditional expectations arising from compact group representations:
if $\pi:G\rightarrow\cl B(H_\pi)$ is a unitary representation of a compact group, then $\cl E:\cl B(H_\pi)\rightarrow\cl B(H_\pi)$ defined by
$$\cl E(T)=\int_G\pi(s)T\pi(s)^*\,dh(s)$$
where $h$ is a normalized Haar measure on $G$, is a conditional expectation onto $\pi(G)'$, so that $N=\pi(G)''$ in this case. A similar class of examples
was considered in \cite{brs}.

\mpar{v4}{conditionally private definition removed, reinstate if desired!}
\mpar{v5}{no desire at this point}

\section{Complementarity with Correctable Subalgebras}\label{s_ccs}

In finite dimensions, a perfect duality exists between operator private and correctable subsystems: a subsystem is correctable for a channel $\cl E$ if and only if it is private
for any complementary channel $\cl E^c$ \cite{kks}. Using the continuity of the Stinespring representation \cite{ksw},
an approximate version of the complementarity theorem was also established \cite{kks}.
In this section, we generalize the notion of complementarity to quantum channels of the form $\cl E:M\rightarrow\cl B(S)$, and in this new framework,
extend the complementarity theorem and its approximate version.

\mpar{v9}{I only made some minor smoothings and corrections to this section.}

\begin{definition}
Let $S$ be a Hilbert space,
let $M$ be a von Neumann algebra, and let $\cl E:M\rightarrow\cl B(S)$ be a quantum channel.
Given a Stinespring triple
$(\pi,V,H)$ for $\cl E$, we define the \emph{complementary channel} of~$\cl E$
with respect to $(\pi,V,H)$ to be
\mpar{v3}{dependence on Stinespring triple noted, and removed later.}
the NUCP map $\cl E^c_{\pi,V,H} : \pi(M)'\rightarrow\cl B(S)$ given by
$$\cl E^c_{\pi,V,H}(X) = V^*XV, \quad X\in\pi(M)'.$$
We also say that $\cl E^c_{\pi,V,H}$ is \emph{a complementary channel of~$\cl E$}.
\end{definition}

\begin{remark}\label{r_newl}\mpar{v3}{Remark added}%
Suppose that $(\pi_1,V_1,H_1)$ and $(\pi_2,V_2,H_2)$ are Stinespring triples
for $\cl E$, and let
$\cl F_1 = \cl E^c_{\pi_1,V_1,H_1}$ and $\cl F_2 = \cl E^c_{\pi_2,V_2,H_2}$.
By the uniqueness of the Stinespring representation,
there exists a partial isometry $U : H_1\to H_2$ satisfying identities (\ref{Stine}).
It follows that, if $Y\in \pi_1(M)'$ and $X\in M$ then
\mpar{v6}{Calculation included.}
\begin{align*}
\cl C_U(Y)\pi_2(X)
& =
UYU^*\pi_2(X) = UY\pi_1(X)U^* = U\pi_1(X)YU^*\\
& =
\pi_1(X)UYU^* = \pi_2(X)\cl C_U(Y);
\end{align*}
thus,
$\cl C_U(\pi_1(M)')\subseteq\pi_2(M)'$ and,
similarly, $\cl C_{U^*}(\pi_2(M)')\subseteq\pi_1(M)'$.
Hence the maps $\cl F_2\circ\cl C_U$ and $\cl F_1\circ\cl C_{U^*}$
are well-defined; by (\ref{Stine}),
$\cl F_1 = \cl F_2\circ\cl C_U$ and $\cl F_2 = \cl F_1\circ\cl C_{U^*}$.
\mpar{v7}{``$\cl E^c$ means any complementary channel''
  deleted (we now say exactly which complements we mean).}%
\end{remark}

\mpar{v7}{There's nothing special about (\ref{ampind2}) here: the same
  argument identifies $\cl E^{cc}$ with $\cl E$ if we use any
  Stinespring rep.~$(\pi,V,H)$ with $\pi$ faithful. Is this worth
  pointing out?}
\mpar{v8}{Done. Text reorganised slightly.}
Let $\cl E : M\rightarrow\cl B(S)$ be a quantum channel and
suppose that $(\pi,V,H)$ is a Stinespring triple for $\cl E$ with $\pi$ faithful.
Let $\cl E^c = \cl E^c_{\pi,V,H}$, and note that
$(\id_{\pi(M)'},V,H)$ is a Stinespring triple for $\cl E^c$
(here $\id_{\pi(M)'} : \pi(M)'\to \pi(M)'$ is the identity map).
Letting $\cl E^{cc} : \pi(M) \to \cl B(S)$ be the complement of $\cl E^c$
with respect to this Stinespring triple, we have that
$$\cl E^c(\pi(X)) = \cl E(X), \ \ \ X\in M.$$
Identifying $M$ with $\pi(M)$, we see that $\cl E^{cc} = \cl E$; thus, the
generalized notion of complementarity is involutive, as expected.

A specific example of a Stinespring triple for $\cl E$
whose corresponding normal representation is faithful can be obtained as follows.
Let $M_1\subseteq\cl B(H_1)$ and $M_2\subseteq\cl B(H_2)$ be
von Neumann algebras. The amplification-induction theorem \cite[Theorem IV.5.5]{t1}
states that for every normal $*$-homomorphism
$\pi$ from $M_1$ onto $M_2$, there exists a Hilbert space $H_3$, a projection $P\in M_1'\oten \cl B(H_3)$
and a unitary $U : H_2\to P(H_1\ten H_3)$ such that
$$\pi(X)=U^*P(X\ten I_{H_3})PU, \ \ X\in M_1.$$
Viewing $PU$ as an isometry $W : H_2\to H_1\ten H_3$, we have
\begin{equation}\label{ampind}
\pi(X)=W^*(X\ten I_{H_3})W, \ \ X\in M_1.
\end{equation}
Now suppose that $M\subseteq\cl B(S')$ is a von Neumann algebra,
$\cl E : M\to\cl B(S)$ is a quantum channel and
$(\pi,V,H)$ is a Stinespring triple for $\cl E$ (with $\pi$ not necessarily faithful).
Since the image $\pi(M)$ is a von Neumann algebra on $H$ \cite{t1},
the amplification-induction theorem allows us to write
\begin{equation}\label{ampind2}
\cl E(X)=\widetilde{V}^*(X\ten I_{H_3})\widetilde{V},\ \ X\in M,
\end{equation}
where $\widetilde{V}=WV$ is the composition of the Stinespring isometry $V:S\rightarrow H$
and the isometry $W : H\rightarrow S'\ten H_3$ from the representation of $\pi$
as in equation (\ref{ampind}).

\mpar{v10}{new par}%
Note that if $M = \cl B(S')$, then $M' = \bb{C}I_{S'}$ and $P = I_{S'}\ten P'\in I_{S'}\oten\cl B(H_3)$
for some $P'\in \cl P(H_3)$,
so we may view $W$ as a unitary from $H$ onto $S'\ten P'H_3$.
Equation (\ref{ampind2}) then becomes the usual Stinespring representation
of a quantum channel $\cl E:\cl B(S')\rightarrow\cl B(S)$,
and its corresponding complement $\cl E^c:I_{S'}\oten\cl B(P'H_3)\to\cl B(S)$
is the usual complementary channel as studied in the literature.

\begin{lemma}\label{l_compcom}
  Let $S$ and $S'$ be Hilbert spaces, $M\subseteq\cl B(S')$ be a von
  Neumann algebra, $\cl E : M\rightarrow\cl B(S)$ be a quantum channel
  and $W : S\to S$ be a partial isometry. 
  If $\cl E^c$ is a complementary channel of $\cl E$, then $\cl
  C_W\circ \cl E^c$ is a complementary channel of $\cl C_W\circ \cl
  E$.
\end{lemma}
\begin{proof}
Suppose that $\cl E^c$ is associated with the Stinespring triple $(\pi,V,H)$ of $\cl E$.
Then
$$\cl E(X) = V^*\pi(X)V, \quad X\in M$$
and
$$\cl E^c(Y) = V^*YV, \quad Y\in \pi(M)'.$$
Thus,
$$\cl C_W\circ \cl E(X) = WV^*\pi(X)VW^*, \quad X\in M,$$
and hence $(\pi, VW^*,H)$ is a Stinespring triple for $\cl C_W\circ \cl E$.
The claim is now immediate.
\end{proof}

Before proceeding to the complementarity theorem,
we recall the operator algebra formalism of quantum error correction \cite{bkk1,bkk1a,bkk2}.

\begin{definition}\label{d_corr}
Let $S$ be a Hilbert space, $M$ be a von Neumann algebra, and
$\cl E:M\rightarrow\cl B(S)$ be a quantum channel.
If $P\in\cl P(S)$, a von Neumann subalgebra $N\subseteq\cl B(PS)$ is said to be
\emph{correctable for $\cl E$ with respect to $P$} if there exists a quantum channel $\cl R:N\rightarrow M$
such that\mpar{v10}{display}%
\[\cl C_P\circ \cl E\circ\cl R=\id_N.\] Given $\varepsilon>0$,
\mpar{v3}{As before, terminology slightly changed}
we say that $N$ is $\varepsilon$-\emph{correctable for $\cl E$ with respect to $P$}
if there exists a quantum channel $\cl R:N\rightarrow M$ such that
$$\Norm{\cl C_P\circ \cl E\circ\cl R-\id_N}_{\cb} < \varepsilon.$$
If $P = I$, we simply say that $N\subseteq\cl B(S)$ is
\emph{correctable} (resp.\ $\varepsilon$-\emph{correctable}) for $\cl E$.
\end{definition}

The above definition unifies the notions of correctable
\mpar{v7}{Should we briefly explain ``noiseless''?}
\mpar{v9}{Done}%
and noiseless (meaning correctable, but with no active correction required) subspaces  and subsystems under one umbrella, allowing
for a general treatment of quantum error correction using the language
of operator algebras. As mentioned in \cite{bkk2}, correctable
subsystems correspond to correctable von Neumann algebras of type~I,
analogous to the situation above for operator private subsystems.

Note that the channel $\cl R$ in Definition \ref{d_corr} (called the
\emph{recovery channel}) has a slightly more general form than the one
usually studied in the literature, (namely, a NUCP map $\cl R:\cl
B(S')\rightarrow\cl B(S)$ satisfying $\cl C_P\circ\cl E\circ\cl R=\cl
C_P|_{N}$).  The reason is to keep in line with our general picture of
quantum channels as NUCP maps whose domain can be a general von Neumann algebra.
\mpar{v3}{Unclear, improve!}
\mpar{v5}{Removed final sentence. Perhaps we can remove the entire
paragraph as it should be clear to the reader.}



\begin{lemma}\label{lemma:corr-complements}
  Let~$M$ be a von Neumann algebra \mpar{v4}{new lemma} and let $\cl
  E:M\to B(S)$ be a quantum channel. If $\epsilon>0$ and $N\subseteq
  \cl B(S)$ is a von Neumann algebra which is $\epsilon$-correctable
  (respectively, correctable) for some particular complement of~$\cl
  E$, then $N$ is $\epsilon$-correctable (respectively, correctable)
  for every complement of~$\cl E$.
\end{lemma}
\begin{proof}
  Let $(\pi_0,V_0,H_0)$ and $(\pi,V,H)$ be Stinespring triples
  for~$\cl E$, and denote the corresponding complements by $\cl E_0^c$
  and $\cl E^c$.  Suppose that $N$ is $\epsilon$-correctable for
  $\cl E_0^c$; we will show that the same is true of~$\cl
  E^c$. There is a quantum channel $\cl R_0:N\to \pi_0(M)'$ with
  $\cbnorm{\cl E_0\circ \cl R_0-\id_N}<\epsilon$. By
  Remark~\ref{r_newl}, there is a partial isometry $U:H_0\to H$ so
  that
\begin{equation}\label{eq_conco}
  U\pi_0(M)'U^*\subseteq \pi(M)',\quad \cl E_0^c=\cl E^c\circ \cl
  C_{U} \ \ \mbox{ and } \ \  \cl E^c=\cl E_0^c\circ \cl C_{U^*}.
\end{equation}
  Fix a normal state
  $\omega\in N_*$ and define a quantum channel $\cl R:N\to \pi(M)'$
  by
  \[ \cl R(T)=U\cl R_0(T)U^*+\langle T,\omega\rangle(1-UU^*),\quad
  T\in N.\] (The second term is required to ensure that $\cl R$ is
  unital.) Since $U^*U$ is a projection, we have $\cl C_{U^*}\circ \cl
  R = \cl C_{U^*}\circ \cl C_{U}\circ \cl R_0$ and so, by (\ref{eq_conco}),
  \[ \cl E^c \circ \cl R =
  \cl E_0^c \circ \cl C_{U^*}\circ \cl R =
 \cl E_0^c \circ \cl C_{U^*}\circ \cl C_{U}\circ \cl R_0 =
 \cl E_0^c \circ \cl R_0.
 \]
 Hence $\cbnorm{\cl E^c\circ \cl R-\id_N}=\cbnorm{\cl E_0^c\circ \cl
   R_0-\id_N}<\epsilon$ and so $N$ is $\epsilon$-correctable for $\cl
 E^c$. The assertion with correctability in place of
 $\epsilon$-correctability is proven by replacing ``less than
 $\epsilon$'' with ``equal to zero'' in the preceding.
\end{proof}

The following elementary lemma will be used to obtain quantum
channels from (not necessarily unital) normal completely positive
maps.

\begin{lemma}\label{lemma:unital}
  Let $S$ be a Hilbert space and
  \mpar{v4}{new lemma}%
  \mpar{v6}{Assumption on the former operator $A$ replaced by the assumption that $\cl F$ be contractive.
  Operator $A$ defined in proof.}%
  \mpar{v6}{Roles of $M$ and $N$ swapped for coherence with the rest of the paper.}%
  let $M$ and $N$ be von Neumann algebras with $N\subseteq \cl B(S)$. If
  ${\cl F} : M\to N$ is a normal completely positive contractive map,
  then there is a quantum channel $\widetilde{\cl F} : M\to
  N$ with $\|\widetilde{\cl F}-\cl E\|_{\cb}\leq 2\|\cl F-\cl E\|_{\cb}$
  for any quantum channel $\cl E : M\to \cl B(S)$.
\end{lemma}
\begin{proof}
  Let $\omega\in M_*$ be a normal state and set $A=1_N-{\cl F}(1_M)$.
  Since $\cl F$ is contractive and positive, ${\cl F}(1_M)$ is a positive contraction and so $A\geq 0$.
  Let $\widetilde{\cl F}$ be the map defined by
  $\widetilde{\cl F}(X)={\cl F}(X)+\langle X,\omega\rangle A$, $X\in M$. Then
  $\widetilde{\cl F}$ is unital by construction, and as it is the sum
  of two normal completely positive maps into~$N$, we see that
  $\widetilde{\cl F}$ is a quantum channel from $M$ into $N$. The map $\widetilde{ \cl
    F}-\cl F$ is completely positive, so it attains its (completely \mpar{v6}{Reference to \cite{Pa} added.}
  bounded) norm at $1_M$ \cite{Pa}; hence, \[ \cbnorm{\widetilde{\cl F}-\cl
    F}=\|\langle 1_M,\omega\rangle A\|= \|A\|=\|(\cl E-\cl
  F)(1_M)\|\leq \cbnorm{\cl F- \cl E}.\]
  Thus $\cbnorm{\widetilde{\cl
      F}-\cl E}\leq \cbnorm{\widetilde{\cl F} -\cl F} + \cbnorm{\cl F-\cl
    E} \leq 2\cbnorm{\cl F-\cl E}$.
\end{proof}

The next theorem is one of the central results of the
paper. \mpar{v3}{commentary added}\mpar{v9}{beefed up referencing note} It generalizes the main results of
both \cite{kks} and \cite{b}, which correspond \mpar{v10}{correpsonds $\to$ correspond}%
to the special case that $S'$ is
finite dimensional and $M = \cl B(S')$.  In the proof, we will use
results from \cite{ksw_JFA}; the latter paper is concerned with the
continuity of the Stinespring representation for completely positive
maps defined on C*-algebras. By Stinespring's theorem for normal maps,
it is straightforward to verify that the results we will need
\mpar{v3}{Sentence added. Leave or delete?}%
\mpar{v4}{Leave}%
remain valid in the case of normal completely positive maps defined on
von Neumann algebras.

\begin{theorem}\label{complsub}
Let $S$ and $S'$ be Hilbert spaces, $M\subseteq\cl B(S')$ be a von Neumann algebra, $\cl E:M\rightarrow\cl B(S)$ be a quantum channel and
$P\in\cl P(S)$.
If a von Neumann subalgebra $N\subseteq\cl B(PS)$ is $\varepsilon$-private (respectively, $\varepsilon$-correctable) for $\cl E$ with respect to $P$
then it is $2\sqrt{\varepsilon}$-correctable (respectively, $8\sqrt{\varepsilon}$-private) for any complement
of $\cl E$ with respect to $P$.
In particular, $N$ is private (respectively, correctable) for $\cl E$ with respect to $P$ if and only if it is correctable (respectively, private)
for any complement of $\cl E$ with respect to $P$.\mpar{v10}{removed $\cl E^c$ notation from statement, seems a bit snappier and clearer to me}
\end{theorem}
\begin{proof}
%
Without loss of generality we may suppose that $P = I$; indeed,
$N\subseteq\cl B(PS)$ is $\ep$-private (respectively,
$\ep$-correctable) for $\cl E$ with respect to $P$ if and only it is
$\ep$-private (respectively, $\ep$-correctable) for $\cl C_P\circ \cl
E$. The general statement now follows from Lemma~\ref{l_compcom},
according to which $\cl C_P\circ\cl E^c$ is complementary to $\cl C_P\circ \cl E$.

We first consider one of the implications in the case $\ep = 0$.
Namely, suppose that $N$ is private for $\cl E$, so that $\cl
E(M)\subseteq N'$, and hence $N=N''\subseteq \cl E(M)'$.
Let $\cl E^c$ be the complement of~$\cl E$ with respect to a minimal
Stinespring triple $(\pi,V,H)$ for $\cl E$.  It follows from
Arveson's commutant lifting theorem~\cite[Theorem 1.3.1]{a} that there
exists a normal *-homomorphism $\rho : \cl E(M)'\rightarrow \pi(M)'$
such that $\rho(X)V=VX$ for all $X\in\cl E(M)'$ (see also~\cite[IV.3.6]{t1}).
Consider the quantum channel $\cl R:=\rho|_N: N\to \pi(M)'$.  Since $\cl E$ is
unital, $V$ is an isometry and hence
$$\cl E^c(\cl R(T))= V^*\rho(T)V = V^*V T = T$$ for all $T\in N$.
Thus, $N$ is correctable for $\cl E^c$. By
Lemma~\ref{lemma:corr-complements}, $N$ is correctable for any
complement of~$\cl E$.


Now suppose that $N$ is $\varepsilon$-private for $\cl E$, so that $N$
is private for some channel $\cl F:M\rightarrow\cl B(S)$ with
$\cbnorm{\cl E-\cl F}<\varepsilon$.  By \cite[Proposition 6]{ksw_JFA},
\mpar{v3}{Reference corrected}%
\mpar{v4}{Why is \cite[Lemma 5]{ksw_JFA} needed here?}%
\mpar{v5}{It's not. Removed}
%
there is a common normal representation
$\pi:M\rightarrow\cl B(H)$ with \mpar{v10}{slight rephrase}%
Stinespring triples
$(\pi, V_{\cl E}, H)$ and $(\pi,
V_{\cl F}, H)$  for $\cl E$ and $\cl F$, respectively, so that
\[\Norm{V_{\cl E}-V_{\cl F}} \leq\sqrt{\cbnorm{\cl E-\cl F}}<\sqrt{\varepsilon}.\]
Let $\cl E^c:\pi(M)'\rightarrow\cl B(S)$ and $\cl
F^c:\pi(M)'\rightarrow\cl B(S)$ be the corresponding
complementary channels.  It follows from \cite[Proposition 3]{ksw_JFA}
that
$$\Norm{\cl E^c-\cl F^c}_{\cb}\leq 2\Norm{V_{\cl E}-V_{\cl F}} < 2\sqrt{\varepsilon}.$$
Since $N$ is private for $\cl F$, it is correctable for $\cl F^c$ by
the previous paragraphs, so there exists a channel
$\cl R : N\rightarrow\pi(M)'$ such that $\cl F^c\circ\cl R=\id_N$. Hence,
$$\Norm{\cl E^c\circ\cl R-\id_N}_{\cb}=\Norm{(\cl E^c-\cl F^c)\circ\cl R}_{\cb} < 2\sqrt{\varepsilon}$$
as $\Norm{\cl R}_{\cb}=1$. Thus, $N$ is
$2\sqrt{\varepsilon}$-correctable for $\cl E^c$. By
Lemma~\ref{lemma:corr-complements}, the same is true of any other
complement of~$\cl E$.

Conversely, suppose that $N$ is $\varepsilon$-correctable for $\cl E$,
so that $\Norm{\cl E\circ\cl R-\id_N} < \ep$ for some quantum channel
$\cl R : N\to M$.  Again by \cite[Proposition 6]{ksw_JFA}, there
exists a common normal representation $\pi : M\rightarrow\cl B(H)$ and
Stinespring triples $(\pi,V_{\cl E\cl R},H)$ and $(\pi,V_{\id},H)$
for $\cl E\circ\cl R$ and $\id_N$, respectively, so that
$$\Norm{V_{\cl E\cl R}-V_{\id}} \leq
\sqrt{\cbnorm{\cl E\circ\cl R - \id_N}} < \sqrt{\varepsilon}.$$
By the amplification-induction theorem,
there exist Hilbert spaces $H_{\cl E}, H_{\cl R}$ and isometries
$V_{\cl E} : S\to S'\ten H_{\cl E}$ and
$V_{\cl R}:S'\to S\ten H_{\cl R}$ such that
$$\cl E(X) = V_{\cl E}^*(X\ten I_{H_{\cl E}})V_{\cl E}, \quad X\in M,$$
and
$$\cl R(T) = V^*_{\cl R}(T\ten I_{H_{\cl R}})V_{\cl R}, \quad T\in N.$$
Thus,
$$\cl E\circ\cl R(T) =
V_{\cl E}^*(V^*_{\cl R}\ten I_{H_{\cl E}})(T\ten I_{H_{\cl R}}\ten I_{H_{\cl E}})(V_{\cl R}\ten I_{H_{\cl E}})V_{\cl E},
\ \ T\in N,$$
and, by Remark \ref{r_newl}, there exists a partial isometry
$U : H\rightarrow S\ten H_{\cl R}\ten H_{\cl E}$ such that
$UV_{\cl E\cl R}=(V_{\cl R}\ten I_{H_{\cl E}})V_{\cl E}$,
$U\pi(T)=(T\ten I_{H_{\cl R}}\ten I_{H_{\cl E}})U$ for all $T\in N$,
and \mpar{v6}{Equation labeled and used later.}
\begin{equation}\label{eq_cust}
\cl C_{U^*}(N'\bar\otimes\cl B(H_{\cl R})\bar\otimes\cl B(H_{\cl E}))\subseteq \pi(N)'.
\end{equation}
Moreover, \mpar{v6}{Equation numbered and used later.}
\begin{equation}\label{eq_neer}
\Norm{(V_{\cl R}\ten I_{H_{\cl E}})V_{\cl E} - UV_{\id}}=
\Norm{UV_{\cl E\cl R} - UV_{\id}} \leq\Norm{V_{\cl E\cl R}-V_{\id}}<\sqrt{\varepsilon}.
\end{equation}

Let $\cl R^c:N'\oten\cl B(H_{\cl R})\to \cl B(S')$ be the complement
of $\cl R$ with respect to the Stinespring triple $(T\mapsto T\otimes
I_{H_{\cl R}},V_{\cl R},S\otimes H_{\cl R})$, and define normal
\mpar{v5}{Changed ``quantum channel'' to ``NCP maps'' as $\cl F$ is not necessarily unital}
completely positive maps
\newcommand{\Nflat}{N'\oten\cl B(H_{\cl R})\oten\cl B(H_{\cl E})}
$\cl F,\cl R^\flat:\Nflat\to \cl B(S)$ by
$\cl F = \cl C_{V^*_{\id}}\circ\cl C_{U^*}$ and
$\cl R^\flat = \cl C_{V^*_{\cl E}}
\circ (\cl R^c\ten\id_{\cl B(H_{\cl E})})$.
By (\ref{eq_neer}) and \cite[Proposition 3]{ksw_JFA},
$\cbnorm{\cl F - \cl R^\flat} < 2\sqrt{\ep}$.
Since $(\pi,V_{\id},H)$ is a Stinespring triple for $\id_N$,
the uniqueness of the Stinespring representation
(see (\ref{Stine})) implies that
there exists a partial isometry $W : H\to S$ satisfying $WV_{\id} = I_S$,
$V_{\id} = W^*$ and
$W\pi(T)=TW$, for $T\in N$. Thus,
$V_{\id}^*\pi(N)'V_{\id}\subseteq N'$ (see Remark \ref{r_newl})
and (\ref{eq_cust}) shows that the image of $\cl F$ lies in $N'$.
By Lemma~\ref{lemma:unital}, there is a quantum channel
\[\widetilde{\cl F}:\Nflat\to N'\quad\text{with}\quad \cbnorm{\widetilde {\cl F}-\cl R^\flat}<4\sqrt\varepsilon.\]
\mpar{v5}{Any ideas on how to manipulate the dilation spaces to remove the additional factor of 2?}%
Since the range of $\cl R$ lies in $M$, we trivially have that $M'$ is private for $\cl R$.
By the first part of the proof, $M'$
is correctable for $\cl R^c$,
so there is a quantum channel
$\cl G : M'\rightarrow N'\oten\cl B(H_{\cl R})$ satisfying $\cl R^c\circ\cl G = \id_{M'}$.
We have \mpar{v6}{Equation numbered and used later.}
\begin{equation}\label{eq_gc}
\cl R^\flat\circ(\cl G\ten\id_{\cl B(H_{\cl E})})=\cl C_{V_{\cl E}^*}|_{M'\oten\cl B(H_{\cl E})}=\cl E^c,
\end{equation}
where $\cl E^c:M'\oten\cl B(H_{\cl E})\rightarrow\cl B(S)$ is
the complement of $\cl E$ with respect to the Stinespring triple $(T\mapsto T\ten I_{H_{\cl E}},V_{\cl E},S'\ten H_{\cl E})$. By (\ref{eq_gc}) and the fact that $\cl G\otimes \id_{\cl B(H_{\cl E})}$ is a complete contraction,
$$\cbnorm{\widetilde{\cl F}\circ(\cl G\ten\id_{\cl B(H_{\cl E})})-\cl E^c}\leq \cbnorm{\widetilde {\cl F}-\cl R^\flat}  < 4\sqrt{\ep}.$$
Since the range of $\widetilde{\cl F}$ is contained in $N'$, the von Neumann algebra $N$ is
$4\sqrt\ep$-private for $\cl E^c$.

Finally, if $\cl E^{\sharp}:\pi^{\sharp}(M)'\rightarrow\cl B(S)$ is another complement to $\cl E$, then there exists a partial isometry
$U^\sharp:H^{\sharp}\rightarrow S'\ten H_{\cl E}$ satisfying $\cl E^{\sharp}=\cl E^c\circ\cl C_{U^{\sharp}}$.
\mpar{v5}{Perhaps if we start with a minimal dilation for $\cl E$, we may take $U^\sharp$ to be an isometry?}
Then
$$\Norm{\widetilde{\cl F}\circ(\cl G\ten\id_{\cl B(H_{\cl E})})\circ\cl C_{U^{\sharp}} - \cl E^{\sharp}}<4\sqrt{\ep}.$$
Applying Lemma~\ref{lemma:unital} to the normal completely positive contraction
$\cl Q = \widetilde{\cl F}\circ(\cl G\ten\id_{\cl B(H_{\cl E})})\circ\cl C_{U^{\sharp}}$, we obtain a quantum channel $\widetilde{\cl Q}:\pi^{\sharp}(M)'\to N'$ satisfying
$\cbnorm{\widetilde{\cl Q}-\cl E^{\sharp}}<8\sqrt{\ep},$
so $N$ is $8\sqrt\ep$-private for $\cl E^\sharp$.
\end{proof}

Applications of Theorem \ref{complsub} to \mpar{v3}{text rewritten slightly}
Gaussian quantum channels will be given in the next section.
In the remainder of the present section, we give two illustrations of this result.
The first one relates to
discrete Schur multipliers; we refer the reader to \cite{Pa} for \mpar{v3}{Reference added.}
the relevant background.
\begin{example}
  Let $X$ be a non-empty countable set and $(\delta_x)_{x\in X}$ be
  the standard orthonormal basis of $\ell_2(X)$.  We identify every
  element of $\cl B(\ell^2(X))$ with its corresponding (possibly
  infinite) matrix $[T_{x,y}]_{x,y\in X}$, where $T_{x,y} = \lan
  T\delta_y,\delta_x\ra$, $x,y\in X$.  Any collection of unit vectors
  $(|\psi_x\ra)_{x\in X}$ in the Hilbert space $H=\ell_2(X)$ defines a
  correlation matrix $C:=[\lan\psi_y|\psi_x\ra]_{x,y\in X}$, which in
  turn yields a NUCP map $\Phi : \cl B(\ell_2(X))\to\cl B(\ell_2(X))$
  via \emph{Schur multiplication}: \mpar{v7}{$\Phi_C$ changed to
    $\Phi$ for simplicity and to avoid $\Phi_C^c$}%
$$\Phi(T) = [\lan\psi_y|\psi_x\ra T_{x,y}]_{x,y\in X}, \ \ T\in\cl B(\ell_2(X)).$$
By abuse of notation, we denote by $\ell_{\infty}(X)$ the von Neumann
subalgebra of diagonal matrices in $\cl B(\ell_2(X))$.  It is
straightforward to verify that $\Phi(D_1TD_2)=D_1\Phi(T)D_2$ for all
$D_1,D_2\in\ell_{\infty}(X)$ and all $T\in\cl B(\ell_2(X))$, {\it
  i.e.}, that $\Phi$ is an $\ell_{\infty}(X)$-bimodule map.  Thus,
$\ell_{\infty}(X)$ is correctable for $\Phi$ and, by Theorem
\ref{complsub}, it is private for any complement $\Phi^c$ of $\Phi$.
In particular, the range of any complement of~$\Phi$ is contained in a
commutative von Neumann algebra, \mpar{v10}{rephrase}%
reflecting the well-known fact that complements of discrete Schur
multipliers are entanglement breaking (see \cite{kmnr}).
\mpar{v9}{Use the Example environment here?}\mpar{v10}{done}%
\end{example}
\mpar{v7}{At the risk of going round in circles, maybe we need $N$ to
  be a $*$-subalgebra of $\cl T(PS)$ rather than a C*-subalgebra of
  $\cl B(PS)$? We seem to want to view $R\in N$ as a functional
  below...}%
We next present an application of Theorem \ref{complsub} by
generalizing the main result in \cite{jk} concerning the structure of
correctable subsystems for finite-dimensional channels as generalized
multiplicative domains.  In \cite[Theorem 11]{jk}, a one-to-one
correspondence was established between correctable subsystems $B$ of a
finite-dimensional channel $\cl E:\cl B(S)\rightarrow\cl B(S)$ and
generalized multiplicative domains $\MD_{\pi}(\cl E)$, where the
latter is defined relative to a projection $P\in\cl P(S)$, a
C*-subalgebra $N\subseteq\cl B(PS)$ and a representation
$\pi:N\rightarrow\cl B(S)$, to be%
\mpar{v5}{This is Definition 2 in~\cite{jk}, except their definition
  does not contain the projection $P$.  Perhaps we should remove the
  dependence on $P$ in $\MD_{\pi}(\cl E)$, and just save it for the
  corollary?}%
\mpar{v7}{R. Without $P$ seems clearer. (The present def is given by
   $\MD_\pi(\cl C_P\circ \cl E)$, anyway).}%
\mpar{v7}{Should we write
  $\MD_\pi(\cl E_*)$ to be more consistent with~\cite{jk}?}%
\mpar{v7}{$S\in N$ changed to~$R\in N$ to avoid notation clash with
  Hilbert space~$S$}%
\mpar{v9}{Notation is fine I think}%
\begin{multline*}
  \MD_{\pi}(\cl E):=\big\{T\in N\mid  \pi(T)(\cl E_*\circ(\cl C_P)_*(R))
= \cl E_*\circ(\cl C_P)_*(TR)
\\\text{and }\hs(\cl E_*\circ(\cl C_P)_*(R))\pi(T)=\cl E_*\circ(\cl C_P)_*(RT), \mbox{ for all } R\in N\big\}.
\end{multline*}
Specifically, if $S=(A\ten B)\oplus(A\ten B)^{\perp}$, then $B$ is
correctable if and only if $I_A\ten\cl B(B)=\MD_{\pi}(\cl E)$ for some
representation $\pi\colon I_A\otimes \cl B(B)\to \cl
B(S)$. 
In the Heisenberg picture, $T\in \MD_{\pi}(\cl E)$ if and only if
\begin{equation*}
\lan\cl (\cl C_P\circ\cl E(X))T,R\ra=\lan\cl C_P\circ\cl E(X\pi(T)),R\ra
\end{equation*}
and
\begin{equation*}
\lan T(\cl C_P\circ \cl E(X)),R\ra=\lan\cl C_P\circ\cl E(\pi(T)X),R\ra
\end{equation*}
for all $R\in N$ and $X\in\cl B(S)$.

\begin{corollary}\label{cor}
Let $S$ and $S'$ be Hilbert spaces, $M\subseteq\cl B(S')$ be a von Neumann algebra, and
$\cl E:M\rightarrow\cl B(S)$ be a quantum channel and $P\in\cl P(S)$. A von Neumann subalgebra $N\subseteq\cl B(PS)$ is correctable for
$\cl E$ with respect to $P$ if and only if there exists a normal
representation $\pi:N\rightarrow M$ such that
\begin{equation}\label{MD2}
(\cl C_P\circ\cl E(X))T=\cl C_P\circ\cl E(X\pi(T))\qand T(\cl C_P\circ\cl E(X)) = \cl C_P\circ\cl E(\pi(T)X)\end{equation}
for all $T\in N$ and $X\in\cl B(S)$.
\end{corollary}

\begin{proof}
As in the proof of Theorem \ref{complsub},
it suffices to consider the case $P = I_S$.
If there exists a normal representation $\pi:N\to M$ satisfying (\ref{MD2}),
then by taking $X = I$ in (\ref{MD2}) and using the fact that
$\cl E$ is unital, we see that $N$ is correctable for $\cl E$.

Conversely, if $N$ is correctable for $\cl E$, then by Theorem \ref{complsub}, $N$ is private
for any complement $\cl E^c$ of $\cl E$.
Taking a Stinespring representation for $\cl E$ of the form
$\cl E(X)=V^*(X\ten I_H)V$, $X\in M$ (see (\ref{ampind2})),
%
%
the corresponding complement $\cl E^c:M'\oten\cl B(H)\to\cl B(S)$ has range in
$N'$, so $N$ is a von Neumann subalgebra of
$\cl E^c(M'\oten\cl B(H))'$. Taking a minimal Stinespring triple $(\pi^c,V^c,H^c)$ for $\cl E^c$,
it follows by Arveson's commutant lifting theorem
\cite[Theorem 1.3.1]{a} that there exists a normal representation $\pi':\cl E^c(M'\oten\cl B(H))'\to\pi^c(M'\oten\cl B(H))'$
satisfying $\pi'(Y)V^c=V^cY$ for all $Y\in\cl E^c(M'\oten\cl B(H))'$.
By the uniqueness of the Stinespring representation, there exists an isometry $W:H^c\rightarrow S'\ten H$
such that $WV^c=V$, $V^c=W^*V$, $W\pi^c(X')=X'W$ for all $X'\in M'\oten\cl B(H)$,
and
$$W\pi^c(M'\oten\cl B(H))'W^*\subseteq(M'\oten\cl B(H))'=M\ten I_H.$$
Let $\pi'':M\ten I_H\to M$ be the $^*$-isomorphism defined by
$\pi''(X\ten1)=X$ and note that, since $W$ is an isometry, $\cl C_W\circ\pi'$ is a
normal *-homomorphism. \mpar{v6}{Pointed out that conjugation by $W$ keeps the homomorphism property.}
Thus,
$\pi:=\pi''\circ\cl C_W\circ\pi'|_{N}:N\to M$ is a normal representation satisfying
\begin{align*}\cl E(X\pi(T))&=V^*((X\pi(T))\ten I_H)V=V^*(X\ten I_H)(\pi(T)\ten I_H)V\\
&=V^*(X\ten I_H)W\pi'(T)W^*V = V^*(X\ten I_H)W\pi'(T)V^c\\
&=V^*(X\ten I_H)WV^cT = V^*(X\ten I_H)VT=\cl E(X)T\end{align*}
for all $X\in M$ and $T\in N$. Similarly, $T\cl E(X)=\cl E(\pi(T)X)$ for all $X\in M$ and $T\in N$.\end{proof}

\begin{remark} Corollary \ref{cor} implies that the correction channel $\cl R$ may always be taken to be $^*$-homomorphism, a
fact previously observed in the case $M=\cl B(S')$ for a separable
Hilbert space $S'$ \cite[Proposition 4.4]{bkk2}.
\end{remark}

\section{Private Algebras for Linear Bosonic Quantum Channels}\label{s_palbc}

In this section we begin our analysis of private algebras and
generalized complementarity for linear bosonic quantum channels,
focusing mainly on the subclass of Gaussian channels. Such channels
arise naturally in the dynamics of open bosonic systems described by
quadratic Hamiltonians (see \cite{wetal} and the references therein).
We begin with a short review of the relevant machinery, adopting the
notation of \cite{h}, to which we refer the reader for details.

Let $\bb{R}^{2n}$ represent the phase space of a system of $n$ bosonic modes. We will write vectors in $\bb{R}^{2n}$ as $z=(x_1,y_1,x_2,y_2,\cdots,x_n,y_n)$,
where $x=(x_1,...,x_n)$ and $y=(y_1,...,y_n)$ are vectors in $\bb{R}^n$ describing the positions and momenta of the $n$ modes.
Let $U,V:\bb{R}^n\to\BLTRN$ be the strongly continuous unitary representations
given by
$$V_x\psi(s)=e^{i\lan x,s\ra}\psi(s) \ \ \text{and} \ \ U_y\psi(s)=\psi(s+y)$$
for $\psi\in\LTRN$ and $s\in\bb{R}^n$. These one parameter groups satisfy the Weyl form of the canonical commutation relations (CCR):
$$U_yV_x=e^{i\lan x,y\ra}V_xU_y, \quad x,y\in\bb{R}^n.$$
Composing the two, we obtain the \emph{Weyl representation}
\mpar{v3}{Meaning of \lq\lq projective''?}%
\mpar{v7}{``projective'' deleted; it seems to mean the representation
  should take values in the quotient group $\cl U(H)/\mathbb T$, but
  that's not the case here}%
\mpar{v10}{Named $W$ and deleted ``unitary representation'' as $W$ isn't one, strictly speaking. (We'd need to quotient out the action of $\bb T$ to get a projective unitary rep, but we're not doing that). Properties of $W$ are presumably well-known anyway.}%
$W:\bb{R}^{2n}\to\BLTRN$ given by
$$W(z)=e^{\frac{i}{2}\lan x,y\ra}V_xU_y, \quad z\in\bb{R}^{2n}.$$
Let
$$\Delta_n = \bigoplus_{i=1}^n\begin{pmatrix} 0 & 1\\ -1 & 0 \end{pmatrix}$$
and, writing $z'=(x_1',y_1',\dots,x_n',y_n')$, let \mpar{v4}{equation corrected}
$$\Delta(z,z') = \langle z,\Delta_n(z')\rangle = \sum_{i=1}^n(x_iy_i'-x_i'y_i)$$
be the canonical \emph{symplectic form} on $\bb{R}^{2n}$.
The Weyl representation $W$ \mpar{v10}{inserted ``Weyl''}%
satisfies the Weyl--Segal form of the CCR:
\begin{equation}\label{CCR}
W(z+z')=e^{\frac{i}{2}\Delta(z,z')}W(z)W(z'), \quad z,z'\in\bb{R}^{2n}.
\end{equation}

The linear transformations $T:\bb{R}^{2n}\to\bb{R}^{2n}$
which preserve the symplectic form~$\Delta$,
in the sense that \mpar{v4}{subscript $n$ removed}%
$$\Delta(Tz,Tz') = \Delta(z,z'),\quad z,z'\in \bb R^{2n},$$
are called \emph{symplectic transformations}. These form a subgroup of
$\GL(2n,\bb{R})$ denoted by $\Sp(2n,\bb{R})$. Note that, by
(\ref{CCR}), $[W(z),W(z')]=0$ if and only if
$\Delta(z,z')\in2\pi\bb{Z}$, where as usual $[X,Y]=XY-YX$ is the
commutator of two operators~$X$ and $Y$.
%
By (\ref{CCR}) and the Stone-von Neumann theorem,
\mpar{v4}{reference for $n$-variable Stone-vN theorem?  Is Theorem~14.8 in [Hall, Quantum
  Theory for Mathematicians] good enough?}
given any $T\in \Sp(2n,\bb{R})$, there exists a unitary $U_T\in\BLTRN$ such
that \mpar{v6}{Labeled and used.}
\begin{equation}\label{eq_stvn}
W(Tz)=U_T^*W(z)U_T
\end{equation}
for all $z\in\bb{R}^{2n}$.\mpar{v10}{Deleted description of the Stone-vN theorem, as we don't discuss projective reps explicity... please check this looks OK!}

An important feature of the Weyl representation $W$ is that it allows
one to study the statistical properties of quantum states via a
``non-commutative characteristic function''.  Specifically, given a state
$\rho\in\TLTRN$, we let 
$\vphi_\rho(z)=\Tr(\rho W(z))$, for $z\in\R^{2n}$.
This characteristic function $\vphi_\rho$ determines the operator
$\rho$ via the following inversion formula:\mpar{v4}{In what sense is
  the integral interpreted?
}%
\mpar{v5}{Added details on weak convergence and an appropriate
  reference.}%
\mpar{v7}{Cor~5.3.4 corrected(?) to Cor~5.3.5---but this only applies
  to Hilbert-Schmidts. What if $\rho\in \cl T(\BLTRN)$ isn't
  Hilbert-Schmidt?}%
$$\rho=\frac{1}{(2\pi)^n}\int_{\bb{R}^{2n}}\vphi_\rho(z)W(-z)\,dz,$$
where the integral converges to~$\rho$ in the weak operator topology
by~\cite[Corollary~5.3.5]{h3}.
A state $\rho\in\TLTRN$ is said to be \emph{Gaussian} if its characteristic function is of the form
$$\vphi_\rho(z)=\exp\left(i\lan m,z\ra-\tfrac{1}{2}\alpha(z,z)\right)$$
where $m\in\bb{R}^{2n}$ is a vector, called the \emph{mean of $\rho$}, and $\alpha$ is a symmetric bilinear form on $\bb{R}^{2n}$ known as the \emph{covariance matrix of $\rho$}.

A \emph{linear bosonic channel} is a quantum channel $\cl
E:\BLTRN\to\BLTRN$ for which there exists $\ell\in \mathbb{N}$, a
state $\rho_E\in\cl T(L_2(\bb{R}^{\ell}))$ in an $\ell$-mode bosonic
environment and a symplectic block matrix
\[ T=\begin{pmatrix} K & L\\ K_E & L_E \end{pmatrix}\in
\Sp(2(n+\ell),\mathbb{R})\] where $K$ is $2n\times 2n$ and $L_E$ is
$2\ell\times 2\ell$, so that if $U_T\in\BLTRN[n+\ell]$ is the unitary
 associated by~(\ref{eq_stvn}) with~$T$, then
the pre-adjoint of~$\cl E$ has the form
\[\cl E_*(\rho)=\Tr_E(U_T(\rho\ten\rho_E)U_T^*),\quad \rho\in\TLTRN\]
where the partial trace is taken over the tensor
factor~$E=\LTRN[\ell]$ of $\LTRN[n+\ell]=\LTRN[n]\otimes \LTRN[\ell]$.
Using the block decomposition of~$T$, one may easily
verify (see \cite[\S12.4.1]{h}) that
\mpar{v5}{Typo: should have been $2n$ and $2m$. Now corrected (with $n=m$) and added reference}%
\mpar{v7}{Thanks, looks good}%
\[\cl E(W(z))=f(z)W(Kz),\quad\text{where}\quad f(z)=\vphi_{\rho_E}(K_Ez),\quad z\in\bb{R}^{2n}.\]
\mpar{v5}{Removed subscript on $W$ as per above. Sentence removed.}%
If $f$ is the characteristic function of a Gaussian state, then $\cl E$ is called a \emph{Gaussian channel}.
In this case, the environment state $\rho_E$ in the representation of $\cl E_*$ is a Gaussian state.

One immediately obtains private subalgebras if $K:\R^n\to\R^n$ does not have full rank. \mpar{v10}{rephrase}%
Indeed, if $R\subseteq\R^n$ denotes the image of $K$,
\mpar{v7}{changed $S$ to $R$ to avoid notation clash below}
then it is clear that
$\cl E(\BLTRN)\subseteq W(R)''$,
where the double commutant $W(R)''$ coincides with
the von Neumann subalgebra of $\cl B(L_2(\bb{R}^n))$
generated by $\{W(z)\mid z\in R\}$.
Let \mpar{v4}{$\Delta_m$ changed to $\Delta$}\mpar{v4}{$m$ changed to $2m$}
\mpar{v5}{Removed subscripts as $n=m$.}
$$R^{\Delta}:=\{z\in\R^{2n}\mid \Delta(z,z')=0 \ \mbox{for all } z'\in R\}$$
be the \emph{symplectic complement} of $R$.
By the CCR (\ref{CCR}),
$[W(z),W(z')]=0$ if 
$\Delta(z,z')=0$,\mpar{v4}{``and only if'' deleted, see above}
and it follows that $W(R^{\Delta})\subseteq W(R)'$.
Hence,
$$\cl E(\BLTRN)\subseteq W(R^{\Delta})'=(W(R^{\Delta})'')',$$
and we have the following result.

\begin{proposition}
Let $\cl E:\BLTRN\to\BLTRN$ be a linear bosonic channel, and let $R$ be the range of the matrix $K$ with
symplectic complement $R^\Delta$. Then the von Neumann algebra $W(R^{\Delta})''$ is private for $\cl E$.
\end{proposition}

\mpar{v9}{Proposition added to have statement stand out}%

\begin{example}\mpar{v10}{example environment added}
  For a simple example with $n=1$, \mpar{v10}{mention that $n=1$}%
  let~$S=\LTR$ and consider the class of single
  mode Gaussian channels $\cl E:\cl B(S)\to\cl B(S)$ satisfying \[\cl
  E(W(z))=f(z)W(Kz),\quad z = (x,y)\in\R^2\] where \[K=\begin{pmatrix}1&0 \\
    0&0\end{pmatrix} \qand
  f(z)=\exp\big(-\tfrac12\alpha(x^2+y^2)\big),\] with $\alpha =
  N_0+\tfrac12$ for some non-negative integer $N_0$.
  \mpar{v3}{Right?}\mpar{v4}{I guess the range of $W$ is weak*-dense
    in $\cl B(L_2)$, so the formula for $\cl E(W(z))$ has a unique
    quantum channel extension $\cl E$ defines on $\cl
    B(L_2)$?}\mpar{v5}{Yes.}%
  This class is known as $A_2$ in Holevo's classification of single
  mode Gaussian channels~\cite{h2}. In this case, the range of~$K$ is
  $R=\R\times\{0\}$ and $R^{\Delta}=R$, so \[W(R)''=\{V_x\mid
  x\in\R\}'' = L_{\infty}(\R)\] is private for $\cl E$, where we
  canonically identify $L_{\infty}(\R)$ with the (abelian) von Neumann
  subalgebra of $\cl B(S)$ consisting of multiplication operators by
  essentially bounded functions.

  By Theorem \ref{complsub}, $L_{\infty}(\R)$ is a correctable
  subalgebra for any complementary channel $\cl E^c$ of $\cl E$.  Let
  us show this explicitly by computing a correction channel~$\cl R$
  for one particular complement $\cl E^c$. First, one may easily
  verify \mpar{v4}{check this once definitions straightened out}%
  \mpar{v5}{Should all be in line}\mpar{v7}{Looks good!}%
  that the pre-adjoint $\cl E_*:\cl T(S)\to\cl T(S)$ can be
  represented as
  \[\cl E_*(\rho)=\Tr_E(U_T(\rho\ten\rho_E)U_T^*), \quad \rho\in\cl
  T(S),\] where $E$ is a copy of~$\LTR$ and $\rho_E\in\cl T(E)$ is the
  Gaussian state with characteristic function $\vphi_{\rho_E}=f$, and
  $T\in \Sp(4,\R)$ is given by the block matrix
  \[ T=
  \begin{pmatrix}
    K&-I\\I&K'
  \end{pmatrix}
  \] where $I$ is the $2\times 2$ identity matrix and $K'=\smallmat{0&0 \\
    0&1}$.

  The state $\rho_E$ is the Gibbs thermal state with mean photon
  number $N_0$, and is pure if and only if $N_0=0$.\mpar{v7}{reference
    for ``Gibbs thermal state'' terminology and purity claim?}  Thus,
  let $E'$ be another copy of~$\LTR$, and let $|\psi\ra\in E\ten E'$
  be a canonical purification of $\rho_E$, that is, $\rho_E =
  \Tr_{E'}(|\psi\ra\lan\psi|)$. \mpar{v3}{syntax corrected}%
  Then
  \[\cl E_*(\rho)=\Tr_{E\ten E'}((U_T\ten
  I_{E'})(\rho\ten|\psi\ra\lan\psi|)(U_T^*\ten I_{E'})), \quad
  \rho\in\cl T(S),\] so we can obtain a complement $\cl E^c:\cl
  B(E\otimes E')\to\cl B(S)$ whose pre-adjoint $\cl E^c_*$ is given by
  \[\cl E^c_*(\rho)=\Tr_{S}((U_T\ten
  I_{E'})(\rho\ten|\psi\ra\lan\psi|)(U_T^*\ten I_{E'})), \quad
  \rho\in\cl T(S).\]
  For $H$ a Hilbert space of the form $L^{2}(\bb{R}^{n})$, let us
  denote \mpar{v6}{$W_H$ explained.}%
  the corresponding Weyl \mpar{v10}{inserted ``corresponding Weyl''}%
  representation $W\colon \bb R^{2n}\to \cl B(H)$ by~$W_H$.  For
  $z,z'\in\R^2$ and $\rho\in\cl T(S)$, we have \mpar{v7}{Calculation
    shortened, hopefully still clear}%
  \begin{align*}
    \langle \cl E^c&(W_{E\ten E'}(z,z')),\rho\rangle\\
    &=\Tr\big((U_T\ten I_{E'})(\rho\ten|\psi\ra\lan\psi|)(U_T^*\ten I_{E'})(I_S\ten W_{E\ten E'}(z,z'))\big)\\
    &=\Tr\big((\rho\ten|\psi\ra\lan\psi|)(U_T^*\ten I_{E'})W_{S\ten E\ten E'}(0,z,z')(U_T\ten I_{E'})\big)\\
    &=\Tr\big((\rho\ten|\psi\ra\lan\psi|)W_{S\ten E\ten E'}(T(0,z),z')\big)\\
    &=\Tr\big((\rho\ten|\psi\ra\lan\psi|)W_{S\ten E\ten E'}(-z,(0,y),z')\big)\\
    &=\Tr\big(|\psi\ra\lan\psi|W_{E\ten E'}((0,y),z')\big)\cdot
    \langle W_S(-z),\rho\rangle.\end{align*} Since $\rho\in\cl T(S)$
  was arbitrary, it follows that
$$\cl E^c(W_{E\ten E'}(z,z'))=\Tr(|\psi\ra\lan\psi|W_{E\ten E'}((0,y),z'))W_S(-z), \quad z,z'\in\R^2.$$
Given the above structure of $\cl E^c$, it is clear that the map
\[\cl R:L_{\infty}(\R)\to\cl B(E\otimes E'),\quad \cl
R(W_S(x,0))=W_{E\ten E'}((-x,0),0),\quad x\in\R\]\mpar{v10}{displayed $\cl R$ definition}%
defines a quantum
channel satisfying $\cl E^c\circ\cl R=\id_{L_{\infty}(\R)}$.
\mpar{v7}{Remark: I guess $\cl R(g)=(x\mapsto g(-x))\otimes I$, $g\in
  L_\infty(\R)$}
\end{example}

\begin{remark}
\mpar{v9}{Seemed appropriate to turn this paragraph into a Remark.}
The symplectic matrix $T$ in the preceding example is not unique. Indeed, any symplectic block matrix of the form
$$\begin{pmatrix} K & \ast\\ I & \ast\end{pmatrix}$$
will do, as only the first column is relevant for the description of $\cl E$. In general,\mpar{v4}{reference for this?}
 \mpar{v5}{Added reference}
 if $A,B:\R^n\to\R^n$ satisfy $\Delta = A^t\Delta A+B^t\Delta B$ (so that the map
$z\mapsto Az\oplus Bz$ is a symplectic embedding), then the matrix
$$\begin{pmatrix} A & \ast\\ B & \ast\end{pmatrix}$$
can be completed to an element of $\Sp(2n,\R)$ (see \cite[Theorem 12.30]{h}). In particular, when $[A,B]=0$, which is the case in the above example, there is a canonical choice for matrices $C,D\in M_n(\R)$ turning
$$\begin{pmatrix} A & C\\ B & D\end{pmatrix}$$
into a symplectic matrix, namely $C=-B'$ and $D=A'$, where
$B'=\Delta^{-1}B^t\Delta$ and $A'=\Delta^{-1}A^t\Delta$ are the \emph{symplectic adjoints} of $A$ and $B$, respectively.
This is precisely how we chose $T$ above, and since the structure of $K'=\smallmat{0&0 \\ 0&1}$ was crucial in determining the recovery channel $\cl R$ (and the overall structure
of $\cl E^c$), the above example may be a glimpse of a deeper connection between complementarity and symplectic duality.
\end{remark}



\section{Conclusion}\label{s_con}

In this paper, we generalized the formalism of private subspaces and
private subsystems to the setting of von Neumann algebras using commutant structures, introduced a generalized
framework for studying complementarity of quantum channels, and established a
general complementarity theorem between operator private and correctable
subalgebras. This new framework is particularly amenable to the important class of linear bosonic channels, and our preliminary investigations suggests a deeper connection between
complementarity and symplectic duality. Moreover, since symplectic geometry has played a decisive role in the development of quantum error correcting codes \cite{crss}, it is
natural to develop such a formalism for private quantum codes via complementarity in both the finite and infinite-dimensional settings.
This, and related questions are currently being pursued and will appear in future work.
\mpar{v9}{At the proof stage I can add in some comments on other finite-dimensional avenues, like the quasiorthogonal algebra direction (and paper) I mentioned to Rupert-Jason, plus the non-operator avenue for the infinite-dim case.}%

\end{document}